\documentclass[11pt]{article}
\usepackage{fullpage}
\usepackage[utf8]{inputenc}
\usepackage[english]{babel}
\usepackage[T1]{fontenc}
\usepackage{amsmath,amssymb,amsfonts,amsthm}
\usepackage{url}
\usepackage{hyperref}
\usepackage{cleveref}
\usepackage{dsfont}

\usepackage{authblk}

\newtheorem{theorem}{Theorem}
\newtheorem*{theorem*}{Theorem}
\newtheorem{lemma}[theorem]{Lemma}
\newtheorem{proposition}[theorem]{Proposition}

\newtheorem{definition}[theorem]{Definition}

\usepackage{authblk}

\newcommand{\sansserif}[1]{%
	\ifmmode
	\mathsf{#1}%
	\else
	\textsf{#1}%
	\fi
}

\newcommand{\R}{\mathbb{R}}
\newcommand{\C}{\mathbb{C}}

\newcommand{\N}{\mathbb{N}}
\newcommand{\cP}{\mathcal{P}}
\newcommand{\cT}{\mathcal{T}}

\newcommand{\cU}{\mathcal{U}}
\newcommand{\cN}{\mathcal{N}}
\newcommand{\cM}{\mathcal{M}}
\newcommand{\cC}{\mathcal{C}}
\newcommand{\cE}{\mathcal{E}}
\newcommand{\cR}{\mathcal{R}}
\newcommand{\cL}{\mathcal{L}}
\newcommand{\cD}{\mathcal{D}}
\newcommand{\fW}{\mathfrak{W}}
\newcommand{\bnb}[1]{\mathcal{B}(#1)}
\newcommand{\den}[1]{\mathcal{D}(#1)}
\newcommand{\bdry}{\partial}

\newcommand{\comp}[1]{\overline{#1}}
\newcommand{\cptp}[1]{\sansserif{CPTP}(#1)}
\newcommand{\sep}[2]{\sansserif{SEP}(#1:#2)}
\newcommand{\sepc}[2]{\sansserif{SEPC}(#1:#2)}

\newcommand{\cp}[1]{\sansserif{CP}(#1)}

\DeclareMathOperator{\hbt}{\mathcal{H}}

\DeclareMathOperator{\tr}{Tr}

\newcommand{\idty}{\mathds{1}}
\newcommand{\idchan}{\mathcal{I}}

\newcommand{\ket}[1]{|#1\rangle}
\newcommand{\dyad}[2]{|#1\rangle \langle #2 |}

\newcommand{\proj}[1]{|#1\rangle \langle #1 |}

\newcommand{\braket}[2]{\langle #1 | #2 \rangle }
\newcommand{\bra}[1]{\langle #1 |}
\newcommand{\fid}[2]{F(#1,#2)}

\newcommand{\chbt}[1]{(\C^2)^{\otimes #1}}

\newcommand{\AL}{{\Lambda_i}}

\newcommand{\ACL}{{\comp{\Lambda_i}}}

\newcommand{\inte}[1]{\sansserif{int}(#1)}
\newcommand{\coh}[2]{I(#1 \rangle #2)}
\newcommand{\mut}[2]{I(#1 : #2)}

\newcommand{\cond}[3]{I(#1 : #2 | #3 )}
\newcommand{\ent}[1]{S(#1)}
\newcommand{\ree}[2]{E_R(#1:#2)}
\newcommand{\relent}[2]{D(#1 \| #2)}

\newcommand{\norm}[1]{\| #1 \|_1}

\DeclareMathOperator{\supp}{\sansserif{supp}}

\newcommand{\dprl}{(\!(}
\newcommand{\dprr}{)\!)}

\newcommand{\dbrl}{[\![}
\newcommand{\dbrr}{]\!]}

\usepackage{xcolor}

\newcommand{\suppress}[1]{}

\title{{\huge A lower bound on the overhead \\
		of quantum error correction in low dimensions}}

\author[1]{Nouédyn Baspin}
\author[2]{Omar Fawzi}
\author[2]{Ala Shayeghi}

\affil[1]{Centre for Engineered Quantum Systems, School of Physics, University of Sydney, New South Wales 2006, Australia}
\affil[2]{Univ Lyon, Inria, ENS Lyon, UCBL,  LIP, F-69342, Lyon Cedex 07, France}

\date{}

\begin{document}
	
	\maketitle
	
	\begin{abstract}
	We show that a quantum architecture with an error correction procedure limited to geometrically local operations incurs an overhead that grows with the system size, even if arbitrary error-free classical computation is allowed. In particular, we prove that in order to operate a quantum error correcting code in 2D at a logical error rate of $\delta$, a space overhead of $\Omega(\sqrt{\log(1/\delta)})$ is needed for any constant depolarizing noise $p > 0$.
	\end{abstract}

	\section{Introduction}
	
	The feasibility of quantum computing relies heavily on finding efficient quantum error correction (QEC) schemes. From a theoretical perspective QEC lies at the heart of the Quantum Threshold Theorem \cite{aharonov1997fault}, and in practice it generally induces costly overheads.
	Part of this cost can be attributed to the necessity of performing frequent measurements to diagnose whether a system has suffered an error. Depending on the architecture considered, those measurements can be challenging to implement, in particular for systems limited to local interactions. The space of observables one has access to is therefore limited by the space that the computer lives in. This observation leads to the following natural question: what is the tradeoff between geometry and the performance of quantum error correction? How much information can reliably be stored in a volume of space?
	
	In this work, we show that an architecture limited to geometrically local operations and classical computation incurs an overhead when using quantum error correction.
	In particular, when limited to \emph{arbitrary} 2D local operations and \emph{free classical computation}, we show that operating a quantum code protecting $k$ logical qubits up to a target error $\delta$, the number of physical qubits $m$ required satisfies
	
	\[
	m \in \Omega\left(k \sqrt{\frac{\log(1/\delta)}{\log(1/p)}}\right) \ ,
	\]
	where $p \in (0,1]$ is the depolarizing noise parameter.
	In most cases, we are interested in an error that decreases exponentially with the system size, or $\delta \sim p^{m^c}$, which gives $\frac{\log(1/\delta)}{\log(1/p)} \sim m^c$. Our bound therefore proves a lower bound on the overhead $m/k \in \Omega(m^{c/2})$.
	In general, for other geometries, our bound reads
	
	\[
	m \in \Omega\left(k \cdot  g_{\text{geom}}\left(\frac{\log(1/\delta)}{\log(1/p)}\right)\right),
	\]
	
	where $g_{\text{geom}}$ depends on the geometry. These bounds differ from existing bounds on \emph{local quantum codes} \cite{bravyi2009no,bravyi2010tradeoffs,delfosse2013tradeoffs,flammia2017limits,baspin2021connectivity} because an architecture with local operations is not necessarily limited to local codes. For example, in \cite{delfosse2021bounds} the authors demonstrate how to measure the syndrome of an arbitrary $n$ qubit sparse code in constant time using $O(n^2)$ ancillas, by making use of free classical computation. Previous attempts to bound the performance of error correction in those systems assumed that only a specific set of gates were allowed, and were limited to a subset of classical communications \cite{delfosse2021bounds}. In this work, our bounds hold for all operations that are separable\footnote{Separable operations are a strict superset of classical operations \cite{chitambar2014everything}.} between the quantum and the classical system. Finally, the methods we use here apply to any architecture that is slow at generating entanglement between its subsystems, and therefore draws a direct connection between one's ability to correct errors, and one's ability to generate entanglement.

	In what follows we review the history of no-go theorems addressing quantum error correction in low-dimensional systems. We will focus on two parameters to capture the performance of a code. First, the dimension $k$ corresponds to the number of qubits protected by the code. Secondly, the distance $d$ is defined as the minimum number of qubits that must be erased for the information to be lost.
	An $\dbrl n,k,d \dbrr$ code is defined on $n$ qubits, has dimension $k$, and distance $d$. An infinite family of codes satisfying $k \in \Omega(n)$ is said to have constant rate.
	
	Previous work addressing the tradeoff between locality and error correction has focused on codes corresponding to the ground space of a sparse frustration-free Hamiltonian. If the terms of this Hamiltonian are spatially local\ \footnote{A code is defined to be local if the terms of its associated Hamiltonian are local}, then given access to nearest neighbors interactions, we can easily verify if a system is in the codespace: it is enough to measure the terms of said Hamiltonian. A celebrated example of such codes is the family of topological codes. In the 20 years since their invention, these codes have seen a sustained theoretical and experimental interest, but no effort could improve their poor parameters $\dbrl n,1, O(\sqrt{n}) \dbrr$. This observation begs the question: are their poor parameters inherent to their locality?
	
	This question was positively resolved in 2009, ten years after Kitaev's surface code, by Bravyi, Poulin, and Terhal (BPT) \cite{bravyi2010tradeoffs}. The authors established that, in 2D, no local code can outperform the surface code: any such code is bound to obey
	\[
	n/k \in \Omega(d^2) \ .
	\]

	This result formalizes a non-trivial constraint on computation with \emph{local codes}. The ratio $n/k$ should here be understood as an overhead, and quantifies the cost of encoding one qubit. Typically, one would like the distance to grow polynomially with $n$, and thus the overhead grows as $n^{\Omega(1)}$. 
	
	It is good here to note that this overhead is not intrinsic to the nature of quantum mechanics, and it is possible to construct significantly better codes when relaxing the assumption of locality. In fact, the same year as the BPT paper, a groundbreaking result introduced a family of codes with constant rate and polynomial distance -- with parameters $\dbrl n,\Theta(n),\Theta(\sqrt{n}) \dbrr$ to be precise \cite{tillich2014quantum}. By plugging these parameters in the BPT bound, one can easily verify that those codes cannot be local. Worse: It is known that if the Hamiltonians corresponding to those codes were to act on qubits placed on a 2D lattice, they would have $\widetilde{\Omega}(n)$ terms spanning a distance $\widetilde{\Omega}(n^{1/4})$ \cite{baspin2021quantifying}. Unfortunately that feature is generic: although constant rate and polynomial distance codes bear the promise of reduced overhead, they all suffer from embarrassingly non-local terms. 
	
	However, all hope is not lost. One might try shuttling qubits around with SWAP gates to emulate theses non-local geometries, incurring some time overhead and additional errors. Another option is to trade this time overhead for a space overhead, by using a large number of ancillas. As previously mentioned, \cite{delfosse2021bounds} demonstrated how to measure the syndrome of an arbitrary $n$ qubit sparse code in constant time using $O(n^2)$ ancillas by using non-local classical computation. Lastly, quantum LDPC codes have very redundant stabilizers, could one operate such code in 2D by only measuring a subset of its stabilizers at a time?  
	
	Presently, the limitations weighing on those alternative options are poorly understood. First it is not clear whether those approaches can be made fault-tolerant and at what cost. Secondly, when deriving bounds on fault-tolerant processes, it is hard to take into account the access to error-free, non-local classical computation. Consequently, we have a rather limited understanding of the resources needed for efficient quantum error correction in low dimensions. In this paper we address this challenge by answering two questions:
	
	\begin{enumerate}
		\item Given access to arbitrary local quantum operations, is it possible to lower bound the overhead of QEC in 2D?
		\item Does this lower bound hold when given access to free classical computation?
	\end{enumerate} 
	
	The framework we use to address those questions provides a new light on the structure of quantum codes, and naturally leads to generalizations and/or stronger versions of known results.
	For example the existing bounds on the encoding/decoding complexity of quantum codes often assume either a unitary circuit \cite{bravyi2006liebrobinson}, or a restricted set of operations. They can also assume a specific structure to the code. For example \cite{aharonov2018quantum} considers a subset of topological codes, \cite{delfosse2021bounds} assume the geometry induced by the stabilizers contains some expansion. Other works have considered the question of the overhead of fault-tolerance, but with no locality restrictions, typically yielding weaker bounds \cite{fawzi2022lower}.
	
	Here we are to eschew those limitations and address the following question: 
	\begin{enumerate}
		\setcounter{enumi}{2}
		\item Given access to arbitrary local quantum operations, and free classical computation is it possible to lower bound the complexity of encoding/decoding quantum codes?
	\end{enumerate}
	We note that quantum circuits assisted by free classical computation can be surprisingly powerful, see e.g., the recent work~\cite{quek2022multivariate} on preparing GHZ states and multivariate trace estimation. We are not aware of any circuit lower bounds for this model.
	
	\subsection{Main results}
	For the sake of readability, we state our result for 2D Euclidean dimensions, although it generalizes to other geometries. Our main contribution reads as follows.
	
	\begin{theorem}[see Theorem~\ref{thm:main-overhead} for the formal version]
		\label{thm:main-informal}
		Let $\cC$ be a code encoding $k$ qubits and let $\fW$ be a 2D-local circuit on $m$ physical qubits, with non-local, error-free classical computation. If $\fW$ is subject to depolarizing noise of strength $p$ every $O(1)$ steps, while achieving a target error below $\delta \equiv p^f$, then 
		\[
		m/k \in \Omega(\sqrt{f}) \ .
		\]
	\end{theorem}

	This result is the first lower bound on the overhead of fault tolerance in low dimensions. It confirms that obtaining exponential error suppression, i.e. $f \sim m^c$ for a constant $c$, incurs polynomial overhead in 2D: $m \in \Omega(k^{2/(2-c)})$. Similarly, constant overhead implies constant error rate, which is in line with previous observations made in \cite{delfosse2021bounds}: the authors measured the syndrome of a constant rate code using $m \sim n \sim k$ qubits, and seem to fail to suppress errors with this scheme. This also shows that the non-local operations used in the construction of constant overhead fault-tolerant schemes~\cite{gottesman2014fault,fawzi2018constant} is necessary.

	As an alternative, one might imagine implementing a constant rate LDPC code $\dbrl N,\Theta(N),D \dbrr$ by concatenating it with a local code of size $n \ $  \footnote{Model suggested by Anirudh Krishna}. Note that although the stabilizers are not local -- and therefore the BPT bound no longer applies -- one could still imagine operating it with local operations. In which case our bound gives $f \in O(n^2)$: the local code has to grow polynomially for exponential error suppression. 
	
	We also obtain a lower bound on the complexity of encoding circuits for any quantum code, and syndrome extracting circuits for any stabilizer code. Those bounds are non-trivial for any constant rate code. 
	
	\begin{theorem} [Encoding circuit depth, see Theorem~\ref{thm:encoding-lower-bound} for the formal version]
		Let $\cC$ be an $\dprl n,k,d \dprr$ quantum code and consider an encoding circuit $\fW$ for $\cC$ on $m$ qubits using arbitrary local operations and free classical computation, then the depth $\Delta$ of $\fW$ obeys
		\[
		\Delta \in \Omega (k\sqrt{d}/m) \ .
		\]
	\end{theorem}

	\begin{theorem} [Syndrome extracting circuit depth, see Theorem~\ref{thm:syndrome-circuit-bound} for the formal version] 
		\label{thm:syndrome-lower-bound-informal}
		Let $\cC$ be an $\dbrl n,k,d \dbrr$ stabilizer quantum code and consider a syndrome extracting circuit $\fW$ for $\cC$ on $m$ qubits using arbitrary local operations and free classical computation, then the depth $\Delta$ of $\fW$ obeys
		\[
		\Delta \in \Omega ({k\sqrt{d}}/{m}) \ .
		\]
	\end{theorem}
	Note how these bounds parallel Theorem \ref{thm:main-informal}.
	The second bound is tight in the regime $m \in O(n)$. In fact, in Section VII of \cite{bacon2015sparse}, it is shown how to measure the syndrome of any LDPC code on $n$ qubits in $O(\sqrt{n})$ time. By taking a good LDPC code \cite{panteleev2021good, leverrier2022quantum} we obtain $\Delta\in \Omega(n \sqrt{n}/n)$, or $\Delta \in \Omega(\sqrt{n})$. Similarly, in the regime $\Delta \in O(1)$, for a good code \cite{panteleev2021good,leverrier2022quantum}, Theorem \ref{thm:syndrome-lower-bound-informal} gives $m \in \Omega(n^{3/2})$, while \cite{delfosse2021bounds} provides a method to do it in $m \in O(n^2)$.

	\subsection{Open questions}
	
	\begin{enumerate}
		
	\item A crucial element in the proof of our main theorem is the observation from Section \ref{subsec:upper-bound} that noisy local circuits have a limited ability to create entanglement. Along the same line of thought one could ask for a precise characterization of entanglement in noisy circuits \cite{dalzell2021random}. As a point of comparison, much has been written regarding the existence an area law in the ground state of local hamiltonians \cite{hastings2007area}, does a similar area law exist for sufficiently deep noisy local circuits? The answer might depend on the metric of choice, but for definiteness, one could ask if for any subset of qubits $\Lambda$ we have $\ree{\Lambda}{\comp{\Lambda}} \in O(\sqrt{|\Lambda|})$ when the circuit runs for more than $\sansserif{polylog}$ time\footnote{Note that for this to hold, we cannot allow classical operations for free: consider $n^{1-\alpha}$ patches of surface code, each of size $ n^\alpha$, for $\frac{1}{2}< \alpha < 1$. Those patches can preserve $\sim n^{1-\alpha}$ Bell pairs for an exponential amount of time.}? This is reminiscent of the fact that noisy \emph{unitary} circuits converge to the maximally mixed state in $\sansserif{polylog}$ time \cite{aharonov1996limitations,muller2016relative}. 
		
		\item Our notion of error rate is quite restrictive: we require the output of the circuit to be close in fidelity to the input state. In practice, for quantum computation, this might not be necessary: for exemple instead of recovering the original state, one might want to measure a logical Pauli observable up to a small error. Can our techniques be adapted to a setting where the definition of the error rate is less restrictive?
		
		\item As discussed in the introduction, local codes can only provide $m/k \sim f^2$, while our bound sits at $m/k \sim \sqrt{f}$, which begs the question: can one actually propose an error correction scheme that is local, yet achieves $m/k \sim \sqrt{f}$ by making use of classical communications, or can our bound be tightened?
		
		\item In the regime $\Delta \in O(1)$, and $m \in O(n)$, our bound Theorem~\ref{thm:syndrome-lower-bound-informal} can be understood as a bound on local stabilizer codes. In that case, we obtain $k\sqrt{d} \in O(n)$, which is far from $k d^2 \in O(n)$ of \cite{bravyi2010tradeoffs}, or even $kd \in \tilde{O}(n)$ of \cite{baspin2021quantifying}. Can one find an intuitive explanation for this difference? In their proofs, \cite{bravyi2010tradeoffs,baspin2021quantifying} only need to focus on two correctable regions, while in this work we typically deal with $m/d$ correctable regions. 
		
	\end{enumerate}
	
	\section{Preliminaries}
	
	For a Hilbert space $\hbt$, $\bnb{\hbt}$ is the space of bounded linear operators on $\hbt$ and $ \den{\hbt}$ is the set of density operators on $\hbt$, i.e., positive semidefinite operators with unit trace. We write $\idty$ for the identity operator in $\hbt$. For a linear operator $\rho$ on $\hbt$, its support is the orthogonal complement of its kernel. We write $\cptp{\hbt_1, \hbt_2}$ the set of completely positive trace-preserving linear maps from $\bnb{\hbt_1}$ to $\bnb{\hbt_2}$ (also called quantum channels) and $\cp{\hbt_1, \hbt_2}$ for completely positive linear maps. When $\hbt_1 = \hbt_2$, we simply write $\cptp{\hbt}$ and $\cp{\hbt}$. We denote $\idchan \in \cptp{\hbt}$ the identity channel on the space of linear operators on $\hbt$.
	
	We are going to consider quantum states and channels that act on multiple systems. The systems will be labelled $A, B, X,\dots$ and should be thought of as labels for a collection of information carrying systems. As such, mathematically, these systems are finite sets. The states of such systems are described by a density operator on the corresponding Hilbert spaces, which we write as $\hbt_{A}, \hbt_{B}, \hbt_{X},\dots$. For example, a state on the systems $AX$ (which should be understood as the union of the sets $A$ and $X$) will be described by an element in $\den{\hbt_{A} \otimes \hbt_{X}}$. We will often include the systems on which the states or channels act as a subscript, e.g., $\rho_{AX} \in \den{\hbt_{A} \otimes \hbt_{X}}$ for the state on $AX$, $\rho_{A} \equiv \tr_{X}(\rho_{AX})$ and $\idchan_{A} \in \cptp{\hbt_{A}}$ the identity channel on the system $A$.
	
	Additionally, we use $\norm{\cdot}$ to denote the trace norm, and $\fid{\rho}{\sigma} = \left(\tr\left(\rho^{\frac{1}{2}} \sigma \rho^{\frac{1}{2}}\right)^{\frac{1}{2}}\right)^2$ for the Uhlmann fidelity \cite{wilde2017quantum}. We use the standard inequalities:
	\begin{align*}
	2(1-\sqrt{\fid{\rho}{\sigma}}) \leq \norm{\rho - \sigma} \leq 2 \sqrt{1-\fid{\rho}{\sigma}} \ ,
	\end{align*}
	and when one of the two states is pure the inequality can be improved to $2(1-\fid{\rho}{\sigma}) \leq \norm{\rho - \sigma}$.
	
	\subsection{Distance and entropic measures}

	One can draw an equivalence between preserving information and preserving entanglement with another party. This is formalized in the following lemma.
	\begin{lemma}[Theorem 2 of \cite{barnum2000quantum}]
		\label{lem:barnum-et-al}
		Let $\cC$ be a subspace of $\hbt$. Let $\cE \in \cptp{\hbt}$ be a quantum channel such that for all states $\ket{\psi} \in \cC$ 
		\[
		\fid{\dyad{\psi}{\psi}}{\cE(\dyad{\psi}{\psi})} \geq 1- \epsilon \enspace.
		\]
		For any state $\rho$ with support included in $\cC$, let $\ket{\rho} \in \hbt_R \otimes \hbt$ be a purification of $\rho$. Then, we have
		
		\[
		\fid{\dyad{\rho}{\rho}}{ (\idchan_R \otimes \cE) (\dyad{\rho}{\rho})} \geq 1-\frac{3}{2}\epsilon \enspace.
		\]
	\end{lemma}

	
	
	We introduce entropic quantities that will be used throughout the paper, in particular the coherent information which is known to capture the ability of a channel to transmit quantum information.
	\begin{definition}
		For a state $\rho \in \den{\hbt}$ and a positive operator $\sigma$ on $\hbt$, the relative entropy is defined as
		\[
		\relent{\rho}{\sigma} \equiv 
		\left\{
		\begin{array}{ll}
		\tr \rho (\log\rho - \log\sigma) & \text{if the support of $\rho$ is included in the support of $\sigma$} \\
		+\infty & {otherwise} \ .
		\end{array}
		\right.
		\]
		The conditional von Neumann entropy of a bipartite state $\rho_{AB} \in \den{\hbt_{A} \otimes \hbt_{B}}$ is defined as
		\[
		\ent{A|B}_{\rho} \equiv - \relent{\rho_{AB}}{\idty_{A} \otimes \rho_B} \ ,
		\]
		which can also be written as $\ent{A|B}_{\rho} = \ent{AB} - \ent{B}$. 		The coherent information is defined as
		\[
		\coh{A}{B}_{\rho} \equiv - \ent{A|B}_{\rho} \ .
		\]		
		The conditional mutual information of a tripartite state $\rho_{ABC} \in \den{\hbt_A \otimes \hbt_B \otimes \hbt_C}$
		\[
		I(A:B|C)_{\rho} \equiv  \ent{A|C}_{\rho} - \ent{A|BC}_{\rho} \ .
		\]						
		
	\end{definition}
		
	Using the monotonicity of the relative entropy (see e.g.,~\cite{wilde2017quantum}) and the continuity statement \cite[Lemma 2]{winter2016tight} for conditional entropy, we immediately obtain the following statements.
	\begin{proposition} Let $\rho_{AB}$ be a bipartite state $\rho_{AB} \in \den{\hbt_{A} \otimes \hbt_{B}}$. The coherent information satisfies the following properties:
		\label{prop:coh-preperties}
		\begin{enumerate}
			\item $\coh{A}{B}_\rho$ is right-monotonous, i.e., for any $\cT \in \cptp{\hbt_B}$, we have $\coh{A}{B}_\rho \geq \coh{A}{B}_{(\idchan_{A}\otimes \cT)(\rho)}$.
			\item Let $\epsilon \in [0,1]$ and $\rho, \sigma$ such that $\fid{\rho}{\sigma} \geq 1-\epsilon$, then $|\coh{A}{B}_\rho - \coh{A}{B}_\sigma| \leq 2 \sqrt{\epsilon} \log \dim \hbt_A + g(\sqrt{\epsilon})$, with $g(\epsilon) = (1+\epsilon)h(\frac{\epsilon}{1+\epsilon})$, where $h(\cdot)$ is the binary entropy function. We use the fact that $g(\epsilon) \leq 2 \sqrt{\epsilon}$.			
		\end{enumerate}
	\end{proposition}
	
	
	Note that the coherent information is not an entanglement measure, in particular $\coh{A}{B}_\rho$ can increase when acting on $A$. We will need to use an entanglement measure: we choose to use the relative entropy of entanglement (REE).
	
	\begin{definition}
		Let $\rho \in \den{\hbt_A \otimes \hbt_B}$, then the relative entropy of entanglement (REE) is defined as  
		\[
		\ree{A}{B}_{\rho} \equiv \min_{\sigma \in \sep{\hbt_{A}}{\hbt_{B}}} \relent{\rho}{\sigma} \ ,
		\]
		where
		\[
		\sep{\hbt_{A}}{\hbt_{B}} \equiv \{\sigma \in \den{\hbt_{A} \otimes \hbt_{B}}: \sigma = \sum_i p_i \sigma_{A,i} \otimes \sigma_{B,i}, \sigma_{A,i} \in \den{\hbt_{A}}, \sigma_{B,i} \in \den{\hbt_{B}}\} 
		\]
		is the set of separable states.
	\end{definition}
	The following proposition summarizes useful properties of the relative entropy of entanglement.
	\begin{proposition} REE satisfies the following properties
		\label{prop:ree-properties}
		\begin{enumerate}
			\item Continuity: 
			let $\rho, \sigma$ such that $\fid{\rho}{\sigma} \geq 1-\epsilon$, then $|\ree{A}{B}_\rho - \ree{A}{B}_\sigma | \leq \sqrt{\epsilon}\log \dim \hbt_A + g(\sqrt{\epsilon})$, with $g(\epsilon) = (1+\epsilon)h(\frac{\epsilon}{1+\epsilon})$.
			\item Monotonicity under separable operations (see Definition~\ref{def:sepc} below): let $\cT $ be a separable quantum channel with respect to the bipartition $A:B$, then $\ree{A}{B}_{\rho} \geq \ree{A}{B}_{\cT(\rho)}$.
			\item $\ree{A}{B} \geq \coh{A}{B}$.
		\end{enumerate}
	\end{proposition}
	\begin{proof}
	The first point can be found in \cite[Corollary 8]{winter2016tight} and the second one in~\cite{vedral1997quantifying}. A proof of the third point is included in the Appendix as Lemma~\ref{lem:coh-lower-bounds-ree}.
	\end{proof}
	
	The set of separable quantum channel is a convenient superset of the set of LOCC operators. We refer to~\cite[Section 6.1.2]{watrous2018theory} for further details on separable quantum channels.
	\begin{definition}
	\label{def:sepc}
		A bipartite quantum channel $\cT\in\cptp{\hbt_{A}\otimes \hbt_{B}}$ is called a separable quantum channel with respect to the bipartition $A:B$ if it admits a Kraus representation with Kraus operators of the form $\{T_{A,i} \otimes T_{B,i}\}_i$. We denote the set of these quantum channels by $\sepc{\hbt_{A}}{\hbt_{B}}$.       
	\end{definition}
	
	\subsection{Circuit model}
	
	Our circuits will have two kinds of systems: a classical system denoted $X$ and quantum systems denoted by the set $A$. As computation on the classical system will be free, we can think of $X$ as a single system, i.e., a set with one element called $X$. On the other hand, $A$ should be seen as a set of qubit systems. Throughout the paper, we will reserve the notation $m$ for the size of $A$. As such $|A| = m = \log \dim \hbt_{A}$.
	In other words, $A$ is interpreted as the set $[m] \equiv \{1, \dots, m\}$.
	In addition for a set $\Lambda \subset A$, $\overline{\Lambda}$ denotes the complement of $\Lambda$ in $A$, i.e., $A \setminus \Lambda$ unless otherwise noted (sometimes, it will denote the complement in a subset of $A$). For any $Q \in \bnb{\hbt_{A}}$, we write $\supp{Q} \subset A$ for the set of qubits on which $Q$ acts non-trivially. Note that even though it shares the same name, $\supp{Q}$ is a subset of $A$ and has nothing to do with the orthogonal complement of the kernel of $Q$. It will always be clear from the context which support we are talking about.

	We consider the following circuit model:
	\begin{definition}
		\label{def:qcircuit}
		A circuit $\fW$ of depth $\Delta$ and width $m$ is a sequence $(\cE_t)_{t=1}^{\Delta}$, with $\cE_t \in \sepc{\hbt_{A}}{\hbt_{X}}$, where $A$ is a set of size $m$ labelling the qubit systems and $\hbt_X$ is an arbitrary finite dimensional Hilbert space. We denote by $[\fW]$ the quantum channel obtained by composing the channels $\cE_t$:
		\[
		[\fW] = \cE_{\Delta} \circ \cdots \circ \cE_1 \ .
		\] 
	\end{definition}
	
	As previously mentioned, the system $A$ should be understood as $m$ qubits and $X$ is introduced to model a classical system that may be used to record and process any classical information, for example obtained from quantum measurements. The distinction between classical and quantum systems is important to allow adaptive quantum circuits with hybrid classical and quantum computations. In fact, in our noise model, we allow the system $X$ to be noise-free as classical computation can be performed with practically perfect accuracy. The circuit width is defined as the number of qubits in the systems $A$ and the system $X$ can have a Hilbert space of arbitrary finite dimension. We also remark that to show our results, we do not need to assume that the system $X$ is classical: the only thing we use is that operations involving $A$ and $X$ have to be separable along the cut $A:X$.

	
	To define the geometric locality of a circuit, we introduce the connectivity graph on the set of qubits $A$.
	\begin{definition}[Connectivity graph]
		Let $A$ be a set of size $m$. A connectivity graph $G = (A, E)$ is an undirected graph on vertex set $A$ and with edge set $E$. 		
		For any $U \subset A$, define $\bdry U = \bdry_+ U \cup \bdry_- U$, where
		\[
		\bdry_- U \equiv \{u \in U : \exists v \in A \setminus U, (u,v) \in E\}
		\]
		\[
		\bdry_+ U \equiv \{u \in A \setminus U : \exists u \in U , (u,v) \in E\} \ .
		\]
	\end{definition}
	
	\begin{definition}
		We say that $\fW = (\cE_t)_{t=1}^{\Delta}$ with qubit systems labelled by the set $A$ is compatible with a connectivity graph $G$ with vertex set $A$ if the following holds: For all $t$, $\cE_t$ admits a Kraus representation with Kraus operators $\{K^{t}_{A,i}\otimes K^{t}_{X,i}\}_i$ with $K^{t}_{A,i} \in \bnb{\hbt_{A}}$ and $K^{t}_{X,i} \in \bnb{\hbt_{X}}$ where $K^t_{A,i} = \bigotimes_{j} K^{t,i}_j$ for some operators $K^{t,i}_j  \in \bnb{\hbt_{A}}$ all satisfying the property: for all $u, v \in \supp(K^{t,i}_{j})$, $(u, v)$ is an edge of $G$.		
	\end{definition}
	For example, for unitary circuits with only two-qubit gates, the operators $K^{t}_{A,i}$ are tensor products of unitary operators on two qubits and the index $i$ can be understood as selecting which two-qubit unitaries to apply as a function of the classical system $X$. The condition of the definition requires that each such two-qubit unitary acts on neighboring vertices in the graph $G$.
	
	Particular connectivity graphs of interest are the ones that can be embedded in a $D$-dimensional Euclidean space where vertices connected by an edge are close in Euclidean distance.
	\begin{definition}
		A connectivity graph $G = (A,E)$ is said to be $c$-local in $D$ dimension if there exists $\eta: A \rightarrow \mathbb{R}^D$ such that 
		\[
		\forall u, v \in A, u\neq v, \| \eta(u)-\eta(v) \|_2 \geq 1
		\]
		and
		\[
		\forall (u,v) \in E, \| \eta(u) - \eta(v) \|_2 \leq c \ .
		\] 
		We will say that $G$ is $O(1)$-local in dimension $D$ if it is $c$-local for some constant $c$ independent of the other parameters in the problem, in particular the number of vertices $|A|$.
		
		We say that a circuit $\fW$ is $O(1)$-local in dimension $D$ if it has a connectivity graph that is $O(1)$-local in dimension $D$.
	\end{definition}

	Next, we establish an important lemma that will be used extensively: it states that applying one step of a circuit with a given connectivity graph can increase entanglement between a region and its complement by at most the size of the corresponding boundary.
	
	\begin{lemma}[Small incremental entangling]
		\label{lem:small-incremental-entangling}
		Let $\fW  = (\cE_t)_t$ be a quantum circuit with a connectivity graph $G$. Then for any $\rho \in \den{\hbt_{A}\otimes \hbt_{X}}$, and every $\cE_t$, we have		
		\[
		\ree{U}{X \comp{U}}_{\cE_t(\rho)} \leq \ree{U}{X \comp{U}}_{\rho} + 3|\bdry U|
		\]		
		for any $U \subset A$, and $\comp{U} \equiv A \setminus U$.
	\end{lemma}
	\begin{proof}
		
		
		For $U \subset A$, let $\tau \in \sep{\hbt_{U}}{\hbt_{\comp{U}X}}$ such that 
		\[
		\ree{U}{X \comp{U}}_{\rho} = \relent{\rho}{\tau} \ .
		\]
		
		Note that, by the fact that $\fW$ has connectivity graph $G$, we can rewrite the Kraus elements $\{\Pi_i\}_i$ of $\cE_t$ as $\Pi_{X,i} \otimes \Pi_{\inte{U},i} \otimes \Pi_{\bdry U,i} \otimes \Pi_{\inte{\comp{U}},i}$
		such that
		$\supp(\Pi_{\inte{U},i})$ which we denote as $\inte{U}_i$ is a subset of $U$, and similarly $\supp(\Pi_{\inte{\comp{U}},i}) \equiv  \inte{\comp{U}}_i \subset \comp{U}$, and $\supp(\Pi_{\bdry U,i}) \equiv {\bdry U}_i\subset \bdry U$. 
		We let $\tau' = \idty_{\bdry U}/2^{|\bdry U|} \otimes \tr_{\bdry U} \cE_t(\tau)$, then note that $\tau' \in \sep{\hbt_{U}}{\hbt_{\comp{U}X}}$. To convince ourselves of this, remember that $\tau$ can be written as
		\[
		\tau = \sum_j p_j \tau_{U}^j \otimes \tau_{\comp{U}X}^j \ .
		\]
		
		Then for any element $\Pi_{X,i} \otimes \Pi_{\inte{U},i} \otimes \Pi_{\bdry U,i} \otimes \Pi_{\inte{\comp{U}},i}$, we have
		\[
		\tr_{{\bdry U}_i } \left( \Pi_{X,i} \otimes \Pi_{\inte{U},i} \otimes \Pi_{\bdry U,i} \otimes \Pi_{\inte{\comp{U}},i} \ \tau \ \Pi_{X,i}^\dagger \otimes \Pi_{\inte{U},i}^\dagger \otimes \Pi_{\bdry U,i}^\dagger \otimes \Pi_{\inte{\comp{U}},i}^\dagger \right) = \sum_j \hat{\tau}^j_{\inte{U}_i}\otimes \hat{\tau}_{\inte{\comp{U}}_iX}^j \
		\]
		for some positive operators $\{\hat{\tau}^j_{\inte{U}_i}\}_j,\{ \hat{\tau}_{\inte{\comp{U}}_iX}^j\}_j$. Now since ${\bdry U}_i \subset \bdry U$, then $\tr_{\bdry U} \cE_t(\tau) \in \sep{\hbt_{U \setminus \bdry U}}{\hbt_{\comp{U}X \setminus \bdry U}}$. We then naturally obtain $\tau' = \idty_{\bdry U}/2^{|\bdry U|} \otimes \tr_{\bdry U} \cE_t(\tau) \in \sep{\hbt_{U}}{\hbt_{\comp{U}X}}$.
		
		Write $\rho' \equiv \cE_t(\rho)$. We will use the fact that for two states $\rho_{CD}, \sigma_{CD} \in \den{\hbt_{C}\otimes \hbt_{D}}$, if $\sigma_{CD} = \sigma_{C} \otimes \sigma_{D}$, then the relative entropy satisfies (Proposition 2 of \cite{capel2018superadditivity})
		
		\[
		\relent{\rho_{CD}}{\sigma_{CD}} = \relent{\rho_C}{\sigma_C}+\mut{C}{D}_\rho + \relent{\rho_D}{\sigma_D} \ .
		\]
		
		Applying this relation to $\rho'$ and $\tau'$, we get
		\begin{align*}
			\ree{U}{\comp{U}X}_{\rho'} & \leq \relent{\rho'}{\tau'} \\ 
			&= \relent{\tr_{\bdry U} \rho'}{\tr_{\bdry U} \tau'} + \mut{\comp{\bdry U} X}{\bdry U}_{\rho'} + \relent{\rho'_{\bdry U}}{\tau'_{\bdry U}} \\
			& \leq \relent{\tr_{\bdry U} \rho'}{\tr_{\bdry U} \tau'} + 3|\bdry U|\\
			& \leq \relent{ \cE_t(\rho)}{ \cE_t(\tau)} + 3|\bdry U|\\
			& \leq \relent{\rho}{\tau} + 3|\bdry U|\\
			& = \ree{U}{X \comp{U}}_{\rho} + 3|\bdry U| \ ,
		\end{align*}
		where we have used the fact that $\mut{\comp{\bdry U} X}{\bdry U}_{\rho'} \leq 2 |\bdry U|$ and $ \relent{\rho'_{\bdry U}}{\tau'_{\bdry U}} \leq |\bdry U|$.
	\end{proof}
	

	%
		%
		%
		%
		%
	
	\subsection{Quantum codes}
	We introduce some basic definitions about quantum error correcting codes. We refer to~\cite{gottesman1997stabilizer} for more details. Let $k, d \leq n$ be positive integers.
	\begin{definition}
		The $n$-qubit Pauli group $\cP_n$ is generated by $\{i,X,Z\}^{\otimes n}$.
	\end{definition}
	\begin{definition}
		A code $\cC$ is a subspace of $(\C^2)^{\otimes n}$ has parameters $\dprl n,k,d \dprr$ if
		\begin{enumerate}
			\item $\cC \cong \chbt{k}$
			\item $\forall \ket{\psi} , \ket{\psi'} \in \cC, \forall P \in \cP_n, |\supp{P}| < d, \bra{\psi} P \ket{\psi'} = c(P) \braket{\psi}{\psi'}$, for some $c(P) \in \mathbb{C}$ \ .
		\end{enumerate}
	\end{definition}
	
	\begin{definition}
		\label{def:stabilizer-codes}
		A code $\cC \subset \chbt{n}$ is said to be a stabilizer code if there exists an Abelian subgroup $S \subset \cP_n$ not containing $-I$, such that	for all $\ket{\psi} \in \hbt$,	
		\[
		\ket{\psi} \in \cC \quad \Leftrightarrow \quad \forall M \in S, M \ket{\psi} = \ket{\psi} \ .
		\]	
		Let $\{M_i\}_{i \in \{1, \dots, n-k\}}$ be independent generators for the group $S$, without loss of generality we take $M_i$ to be Hermitian. For any state in $\hbt$, we write $s_i \in \{-1,+1\}$ the outcome of the measurement of $M_i$. The vector $s = (s_i)_{i \in \{1, \dots, n-k\}}$ is called the syndrome of this state.
		The Hilbert space then splits as a direct sum of syndrome subspaces: $\hbt = \bigoplus\limits_{s \in \{-1,+1\}^{n-k}} \cC_s$, where $\cC_s$ is defined by $\ket{\psi} \in \cC_s  \Leftrightarrow \forall i \in \{1,\dots, n-k\}, M_i \ket{\psi} = s_i \ket{\psi}$. For stabilizer codes, we use the notation $\dbrl n, k, d \dbrr$ when the subspace $\cC$ has dimension $2^k$ and minimum distance $d$.
	\end{definition}
	
	\begin{definition}
		Let $\cC$ be a code. Then a region $\Lambda \subset [n]$ is said to be correctable if there exists $\cR_\Lambda$ such that $\forall \rho \in \cC, \cR_\Lambda \circ \tr_\Lambda (\rho) = \rho$.
	\end{definition}

	
	The following standard lemma shows that any region with size at most $d-1$ is correctable.
	\begin{lemma}
		\label{lem:less-than-d-corr}
		Let $\cC$ be a code on $n$ qubits. Then any region $\Lambda \subset [n]$ with $|\Lambda| < d$ is correctable.
	\end{lemma}
		%
		%
		%
		%
	
	The next lemma shows that for any state in the code, the reduced state on a correctable region is independent of the code state. This even holds in an approximate sense.
	
	%
	%

	\begin{lemma}[Approximate indistinguishability]
		\label{lem:approx-indistinguishability}
		Let $\epsilon \in [0,1]$, $A'$ be an $n$-qubit system and let $\cC$ be a code such that for any region $\Lambda \subset A'$ of size $|\Lambda| < d$ there exists $\cR$ such that for all $\rho$ supported on $\cC$,
		we have	
		\[
		\fid{\cR \circ \tr_\Lambda (\rho) }{\rho} \geq 1-\epsilon \ .
		\]
		Then there exists $\omega_{\Lambda} \in \den{\hbt_{\Lambda}}$ such that for any state $\rho$ supported on $\cC$, the following is satisfied
		\[
		\fid{\omega_{\Lambda}}{\rho_{\Lambda }} \geq 1-\frac{3 \epsilon}{2} \ .
		\]
	\end{lemma}
	
	\begin{proof}
		From the recovery condition, and Lemma \ref{lem:barnum-et-al}, we can verify that for any purification $\ket{\rho}_{A'R}$ of $\rho_{A'}$ satisfies
		\[
		\fid{\idchan_R \otimes \cR \circ \tr_\Lambda (\proj{\rho})}{\proj{\rho}} \geq 1-\frac{3}{2}\epsilon \ .
		\]
		Then from Theorem 3 of \cite{flammia2017limits}, we can verify that there exists a state $\omega_{\Lambda}$ such that for all states $\rho_{A'R}$ the following is satisfied
		\[
		\sqrt{1-\fid{\omega_\Lambda\otimes \rho_R}{\rho_{\Lambda R}}} \leq \sqrt{\frac{3\epsilon}{2}} \ .
		\]
		This gives $\fid{\omega_{\Lambda}\otimes \rho_R}{\rho_{\Lambda R}} \geq 1-\frac{3 \epsilon}{2}$, and by the monotonicity of the fidelity we obtain $\fid{\omega_{\Lambda}}{\rho_{\Lambda }} \geq 1-\frac{3 \epsilon}{2}$.
	\end{proof}

	\section{Lower bounds for error correction in low dimensions}
	
	In this section, we establish lower bounds on the size of geometrically local circuits for preparing a code state for a quantum code with a large minimum distance and measuring the syndrome of such a stabilizer code. 
	To define a quantum circuit that implements such tasks, we have to choose a subset of qubits $A' \subset A$ that contain the desired outcome. Recall that $A$ denotes the set of all $m$ qubits used by the circuit and $A'$ will be smaller, typically of size $n$. 
		
	\subsection{Entropic properties for code states}

	\begin{lemma}
		\label{lem:corr-is-max-entangled}
		Let $\cC$ be a $ \dprl n,k,d \dprr$ code and $A'$ be a set of size $n$ labelling $n$ qubits. Then for any region $\Lambda \subset A'$ such that $|\Lambda| < d$, and for any state $\rho \in \den{\hbt_{A'}}$ that has a support included in $\cC$, we have 
		\[
		\coh{\Lambda}{\comp{\Lambda}}_\rho = \ent{\Lambda}_\rho \ ,
		\] 
		where $\comp{\Lambda} \equiv A' \setminus \Lambda$.
	\end{lemma}
	
	\begin{proof}
		Let $\ket{\rho} \in \hbt_{R} \otimes \hbt_{A'}$ be a purification of $\rho_{A'}$. We write $\rho_{RA'} = \dyad{\rho}{\rho}$. As Lemma \ref{lem:less-than-d-corr} guarantees the existence of $\cR$ a recovery map, we have from Lemma \ref{lem:barnum-et-al} that  $\fid{\rho_{RA'}}{\idchan_R \otimes \cR \circ \tr_\Lambda(\rho_{RA'})} =1$. From the right-monotonicity of the coherent information, we have
		\[
		\coh{R}{A'}_{\rho} \geq \coh{R}{\comp{\Lambda}}_{\rho} \geq \coh{R}{A'}_{\idchan_R \otimes \cR (\rho_{R\comp{\Lambda}})} = \coh{R}{A'}_{\rho} \ .
		\]
		One can then verify that $\coh{R}{\comp{\Lambda}}_{\rho} = \coh{R}{A'}_{\rho}$ can be rewritten as $\ent{\Lambda}_{\rho} = \coh{\Lambda}{\comp{\Lambda}}_{\rho}$.
	\end{proof}

	We can then show that 
	
	\begin{lemma}
		
		\label{lem:structure-code}
		Let $\cC$ be a $\dprl n,k,d \dprr$ code and $A'$ be a set of size $n$ labelling $n$ qubits. Then for any partition $\{\Lambda_i\}_i$ of $A'$ such that $|\Lambda_i| < d$, we have, for any state $\rho \in \den{\hbt_{A'}}$ that has a support included in $\cC$
		\[
		\sum_i \ree{\Lambda_i}{\comp{\Lambda_i}}_\rho \geq k \ ,
		\] 
		where $\comp{\Lambda} \equiv A' \setminus \Lambda$.
	\end{lemma}
	
	\begin{proof}
		We have, from Lemma \ref{lem:corr-is-max-entangled}
		\[
		\sum_i \coh{\Lambda_i}{\comp{\Lambda_i}}_\rho = \sum_i \ent{\Lambda_i}_\rho \ .
		\]		
		Write $\sigma = 2^{-k} \Pi_\cC$ with $\Pi_\cC$ the projector on $\cC$. Since every $\Lambda_i$ satisfies $|\Lambda_i| < d$, we can use the indistinguishability of quantum codes Lemma \ref{lem:approx-indistinguishability} (with $\epsilon = 0$), and we have $\rho_{\Lambda_i} = \sigma_{\Lambda_i}$. Further, the coherent information lower bounds the REE  (Proposition \ref{prop:ree-properties}). This gives:	
		\[
		\sum_i \ree{\Lambda_i}{\comp{\Lambda_i}}_\rho  \geq \sum_i \coh{\Lambda_i}{\comp{\Lambda_i}}_\rho = \sum_i \ent{\Lambda_i}_{\sigma} \geq \ent{A'}_{\sigma} = k \ ,
		\]
		where the last inequality stems from the subadditivity of the entropy.
	\end{proof}

	\subsection{Lower bounds in terms of the minimum distance}

	In this part, we show how Lemma \ref{lem:structure-code} implies lower bounds on the depth of circuits preparing code states. For simplicity of exposition, in what follows we focus on $D$-dimensional Euclidean spaces though this can be generalized to more complicated geometries.
	
	In what follows, $\proj{0}_{X}$ denotes a fixed pure state in $\hbt_{X}$ and $\proj{0}_{A'}$ for $A' \subset A$ denotes the product state $\proj{0}$ on all qubits of $A'$.
	
	\begin{theorem} [Encoding circuits]
		\label{thm:encoding-lower-bound}
		Let $\cC$ be an $\dprl n,k,d \dprr$ quantum code. Let $\fW$ be a $D$-dimensional $O(1)$-local quantum circuit of depth $\Delta$ and width $m \geq n$, and let $A' \subset A$ be a subset of the qubits of $\fW$ of size $n$. Assume that the output state $\rho_{AX} = [\fW](\proj{0}_{A} \otimes \proj{0}_X)$ of the circuit is such that $\rho_{A'}$ is fully supported on $\cC$. Then, we have 
		\[
		\Delta \in \Omega \left(\frac{kd^{1/D}}{m}\right) \ .
		\]
	\end{theorem}
	We note that the proof shows more generally that for any circuit $\fW$ with connectivity graph $G$ and any partition $\{\Gamma_i\}_{i=1}^{\ell}$ of $A$, we have $\Delta \geq \frac{k}{3\sum_{i=1}^{\ell} |\bdry \Gamma_i|}$.   
	\begin{proof}
		On one hand, by Lemma \ref{lem:structure-code}, any partition $\{\Gamma_i\}_{i=1}^{\ell}$ of $A$ such that $|\Gamma_i| < d$ induces a partition $\{\Lambda_i\}_{i=1}^{\ell}$ of $A'$, with $\Lambda_i = \Gamma_i \cap A'$, and $|\Lambda_i| < d$. We can then guarantee that there exists $i'$ such that
		\[
		\ree{\Lambda_{i'}}{ A' \setminus \Lambda_{i'} }_{\rho}  =  \ree{\Lambda_{i'}}{A' \setminus \Lambda_{i'}}_{\rho} \geq k/\ell \ .
		\]
		
		On the other hand, the Small Incremental Entangling Lemma \ref{lem:small-incremental-entangling} guarantees that for any $\Gamma_i$ we have
		\begin{align*}
			\ree{\Gamma_i}{\comp{\Gamma_i}X}_{\rho} &= \ree{\Gamma_i}{\comp{\Gamma_i}X}_{\cE_{\Delta} \circ ... \circ \cE_1 (\dyad{0}{0}_{A} \otimes \dyad{0}{0}_{X})} \\
			&\leq 3 \Delta |\bdry \Gamma_i| + \ree{\Gamma_i}{\comp{\Gamma_i}X}_{\dyad{0}{0}_{A} \otimes \dyad{0}{0}_{X}} \\
			& = 3 \Delta |\bdry \Gamma_i| \ ,
		\end{align*}
		where we recall the notation $\comp{\Gamma_i} = A \setminus \Gamma_i$ and that $\dyad{0}{0}_{A}$ refers to the product state where all qubits of $A$ are set to $\proj{0}$.
		
		By the monotonicity of the REE under separable operations, we can now obtain
		\begin{align*}
			3\Delta|\bdry \Gamma_{i'}| &\geq \ree{\Gamma_{i'}}{\comp{\Gamma_{i'}}X}_{\rho}  \\
			&\geq \ree{\Lambda_{i'}}{A' \setminus \Lambda_{i'} }_{\rho}  \\
			&\geq k/\ell \ .
		\end{align*}
		
		Since the circuit is $D$-dimensional, there always exist a partition $\{\Gamma_i\}_{i=1}^{\ell}$ of a $D$-dimensional graph such that $|\Gamma_i| \leq \lambda$,  $|\bdry \Gamma_i| \in O(\lambda^{D-1/D})$, $\ell \in O(m/\lambda)$, for any $\lambda$, see Lemma \ref{lem:geometric-partitioning}. Picking $\lambda = d-1$ and applying the inequality we obtained previously, we have $O(\Delta \lambda^{(D-1)/D}) \geq k \lambda /m$, or $\Delta \in \Omega(\frac{k \lambda^{1/D}}{m})$. Since $\lambda = d-1$, we obtain the desired result.
	\end{proof}
	
	
	Next, we move to the problem of syndrome extraction for stabilizer codes.
	\begin{theorem} [Syndrome extracting circuit]
		\label{thm:syndrome-circuit-bound}
		Let $\cC$ be a $\dbrl n,k,d \dbrr$ stabilizer code. Assume that $A' \subset A$ and $\fW  = (\cE_t)_{t=1}^{\Delta}$ is a circuit such that for any $\rho \in \den{\hbt_{A'}}$ we have
		\[
		\tr_{\comp{A'}} \circ [\fW] \left(\rho_{A'} \otimes \dyad{0}{0}_{\comp{A'}} \otimes \proj{0}_{X} \right) = \sum_{s} \Pi_s \rho \Pi_s \otimes \dyad{s}{s}_{X} \ ,
		\]
		where $\Pi_s$ is the projector onto the syndrome subspace $\cC_s$.
		Then $\Delta$ obeys
		\[
		\Delta \in \Omega \left(\frac{k d^{1/D}}{m}\right) \ .
		\]
	\end{theorem}
	
	\begin{proof}
		%
		
		The Hilbert space on $n$ qubits naturally splits as $ \hbt = \bigoplus_{s \in \{-1,+1\}^{n-k}}  \cC_s$. Applying the circuit $\fW$ to the state $\proj{0}^{\otimes m}$, we obtain after tracing out $\comp{A'}$ the state $\sum_{s \in \{-1,+1\}^{n-k}} \Pi_s \proj{0}^{\otimes n} \Pi_s \otimes \proj{s}$ by assumption. Note that for any $s \in \{-1,+1\}^{n-k}$, there exists an operator $P_{s} \in \cP_n$ (which can be seen as a correction operator for error syndrome $s$) such that for any $\ket{\psi} \in \cC_s$, we have $P_{s} \ket{\psi} \in \cC$. 
		We add one step to this circuit: a recovery operation with Kraus elements $\{P_s \otimes \ket{0}\bra{s}\}_{s \in \{-1,+1\}^{n-k}}$. Note that this map is in $\sepc{\hbt_{A}}{\hbt_{X}}$ and as a result we obtain a circuit of depth $\Delta + 1$. The state obtained on the register $A'$ is then $\sum_{s} P_s \Pi_s \proj{0}^{\otimes n} \Pi_s P_s$. As $P_s \Pi_s \ket{0}^{\otimes n} \in \cC$ for any $s$, this state is supported on $\cC$ and we can apply Theorem \ref{thm:encoding-lower-bound}:
		\[
		\Delta + 1 \in \Omega(k d^{1/D} /m) \ .
		\]	
		%
		%
		%
	\end{proof}

		%
	%
	%

	\section{Lower bounds for error correction for noisy circuits}

	In this section, instead of making an assumption on the minimum distance of the code, we make an assumption on the logical error rate that is achieved by the error correction module. For a given noise model, we say that a circuit defining an error correction module has logical error rate $\delta$ if after the (ideal) error correction module is applied, the output remains $\delta$-close to the correct state. We show that for any quantum code with an error correction module that is geometrically local, the memory overhead has to grow when the desired logical error rate decreases.

	For this section, it is convenient to describe a code $\cC$ by an encoding isometry $U : \chbt{k} \to \chbt{n}$, i.e., $\cC = \mathrm{Im}(U)$. We also introduce the systems $R$ and $L$ corresponding to $k$ qubits and let $\Phi_{RL} \in \den{\hbt_{R} \otimes \hbt_{L}}$ be a maximally entangled state. In addition let $\cU \in \cptp{\hbt_{L}, \hbt_{A'}}$ be the encoding quantum channel that maps the logical information to the code space:
	\[
	\cU( \cdot ) \equiv U \cdot U^{\dagger} \ .
	\]
	We also define the preparation map $\cP \in \cptp{\mathbb{C}, \hbt_{\comp{A'}} \otimes \hbt_{X}}$ as 
	\[
	\cP( \cdot ) \equiv \tr(\cdot) \proj{0}_{\comp{A'}} \otimes \proj{0}_{X} \ .	
	\]
	\begin{definition}
		\label{def:ec-module}
		An error-correction module for the code defined by the isometry $U$ is a family of circuit $(\fW_j)_{j=1}^{J}$ with $\fW_{j} = (\cE_{t,j})_{t=1}^{\Delta}$ all acting on the same systems $AX$, and a choice of subset $A' \subset A$ of size $n$.  We say that such a module has logical error rate $\delta$ if 
		\begin{align*}
			F\left( \idchan_{R} \otimes \left( \tr_{\comp{A'}X} \circ [\fW]_{p} \circ (\cU \otimes \cP) \right)(\Phi_{RL}), 
			\idchan_{R} \otimes \cU(\Phi_{RL}) \right) \geq 1-\delta \ ,
		\end{align*}
		where $[\fW]_p$ is the map obtained by applying noise before each circuit $\fW_j$ and composing all the circuits:
		\[
		[\fW]_{p} = \bigcirc_{j=1}^{J}(\cE_{\Delta,j} \circ ... \circ \cE_{1,j} \circ (\cN_p^{\otimes m} \otimes \idchan_{X}))_i \ ,
		\]	
		with $\cN_p$ the $p$-depolarizing channel defined as
		\[
		\cN_p(\rho) = (1-p)\rho + p\tr(\rho) \idty/2 \ .
		\]
		The number $\Delta$ is called the depth of the error correction module.
	\end{definition}
	
	\subsection{Lower bound on entanglement}
	
	In this section we show how the existence of a good error-correction module implies that the codestates of $\cC$ are highly entangled. The following lemma can be thought of as analogous to Lemma~\ref{lem:structure-code} where instead of imposing a constraint on the minimum distance of the code, we assume that the code can correct errors with good accuracy.
	\begin{lemma}
		\label{lem:structure-unitary}
		Using the same notation as in the paragraph preceding Defintion~\ref{def:ec-module}, assume there exists a decoding map $\cD \in \cptp{\hbt_{A} \otimes \hbt_{X}, \hbt_{A'}}$ such that
		\begin{align*}
			F\left( \idchan_{R} \otimes \left( \cD \circ (\cN^{\otimes m}_p \otimes \idchan_{X}) \circ (\cU \otimes \cP) \right)(\Phi_{RL}), 
			\idchan_{R} \otimes \cU(\Phi_{RL}) \right) \geq 1-\epsilon \ .
		\end{align*}
		Then for any partition $\{\Lambda_i\}_i, \Lambda_i \subset A'$ of $A'$, we have
		\[
		\sum_i \ree{\Lambda_i}{\ACL}_{\idchan_R\otimes \cU (\Phi_{RL})} \geq k - \sum_i 2\sqrt{{\epsilon}/{p^{|\Lambda_i|}}}|\Lambda_i| + g(\sqrt{{\epsilon}/{p^{|\Lambda_i|}}}) \ ,
		\]
		where $\comp{\Lambda_i} = A' \setminus \Lambda_i$.
	\end{lemma}
	\begin{proof}
		We can write $\cN_p^{\otimes m} = p^{|\Lambda_i|}\cN_{\Lambda_i} + (1-p^{|\Lambda_i|})\cM_{i}$, with $\cN_{\Lambda_i} = (\idty_{\Lambda_i}/2^{|\Lambda_i|}\tr_{\Lambda_i})\otimes \cN^{\otimes |\comp{\Lambda_i}|}$ and $\cM_i$ some quantum channel. 
		
		Using the assumed bound on the fidelity together with Lemma \ref{lem:convex-is-close}, 
		we get
		\begin{align}
		\label{eq:bound_fid_lem_lb_ent}
		F\left( \idchan_{R} \otimes \left( \cD \circ (\cN_{\Lambda_i} \otimes \idchan_{X}) \circ (\cU \otimes \cP) \right)(\Phi_{RL}), 
			\idchan_{R} \otimes \cU(\Phi_{RL}) \right) \geq 1-\epsilon/{p^{|\Lambda_i|}} \ .
		\end{align}
		
		Let us define the state $\tau_{RA'X} = \idchan_R \otimes \cU (\Phi_{RL}) \otimes \proj{0}_{X}$. Then we have
		\begin{align*}
			\ent{ \Lambda_i }_{\tau} = \coh{\Lambda_i}{R\ACL}_{\tau} &= \coh{\Lambda_i}{\ACL}_{\tau} + I(\Lambda_i:R|\ACL)_{\tau} \\
			&= \coh{\Lambda_i}{\ACL}_{\tau} + I(\Lambda_i:R|X\ACL)_{\tau} \ .
		\end{align*}
		In order to upper bound, $I(\Lambda_i:R|X\ACL)_{\tau}$, we use the inequality \eqref{eq:bound_fid_lem_lb_ent}. For that, consider the recovery channel $\cR \in \cptp{\hbt_{\comp{\Lambda_i}} \otimes \hbt_{X}, \hbt_{A'} \otimes \hbt_{X}}$ defined by $\cR(\omega_{\comp{\Lambda_i} X}) = \cD(\idty_{\Lambda_i}/2^{|\Lambda_i|} \otimes \cN_p^{\otimes |\comp{A'}|}(\proj{0}_{\comp{A'}}) \otimes \cN_p^{\otimes |\Lambda_i|}(\omega_{\comp{\Lambda_i} X})) \otimes \proj{0}_{X}$. Then it is easy to see that 
		\[
		F(\idchan_{R} \otimes \cR(\tr_{\Lambda_i} \circ (\cU \otimes \cP)(\Phi_{RL})), \tau_{RA'X}) \geq 1-\epsilon/p^{|\Lambda_i|} \ .
		\]
		By Lemma \ref{lem:approximate-markov}, we obtain
		\[
		\ent{ \Lambda_i }_{\tau} \leq \coh{\Lambda_i}{\ACL}_{\tau} + 2\sqrt{{\epsilon}/{p^{|\Lambda_i|}}}|\Lambda_i| + g(\sqrt{{\epsilon}/{p^{|\Lambda_i|}}}) \ .
		\]
		 Equivalently 
		\[
		\coh{\Lambda_i}{\ACL}_{\tau} \geq \ent{\Lambda_i}_{\tau} - 2\sqrt{{\epsilon}/{p^{|\Lambda_i|}}}|\Lambda_i| - g(\sqrt{{\epsilon}/{p^{|\Lambda_i|}}}) \equiv  \ent{\Lambda_i}_{\tau}  - h_{|\Lambda_i|} \ .
		\]
		
		However, as it stands this bound is not very restrictive, as $\ent{\Lambda_i}_{\tau}$ could take any value. To resolve this issue, we sum over the individual contributions, which yields
		\begin{align*}
			\sum_i \coh{\Lambda_i
			}{\comp{\Lambda_i}}_{\tau} &\geq \sum_i \ent{\Lambda_i
			}_{\tau} - h_{|\Lambda_i|}\\ &\geq k- \sum_i h_{|\Lambda_i|} \ .
		\end{align*}
		Since $\sum_i \ent{\Lambda_i
		} \geq \ent{A'} = k$. The result then follows from Lemma \ref{lem:coh-lower-bounds-ree}.
	\end{proof}

	\subsection{Upper bound on entanglement}
	\label{subsec:upper-bound}
	
	We now describe how a local noisy circuit is limited in its ability to generate highly entangled states. We will later leverage this element in our proof of our main theorem: to preserve information, the circuit needs to produce entangled states, which it cannot, due to its locality.
	
	More specifically, this lemma formalizes the following: as the system is affected by noise, a region $\Gamma \subset A$ can only recover $O(|\bdry\Gamma|)$ qubits of entanglement between any two noise layers. Naturally, this lower bounds the ability of any $\Lambda \subset \Gamma$ to be entangled with the rest of the system. 
	
	The motivation behind distinguishing $\Lambda$ and $\Gamma$ is the following. Later, we will take part of $\cC$ to live on $\Lambda$, which is therefore entangled with the rest of the system. However it is hard to partition the data qubits in a manner that guarantees a small boundary to each region, as we have no information regarding how they are arranged vis-a-vis the ancillary qubits. Indeed, it is easier to partition the whole system, and $\Lambda \subset \Gamma$ inherit the $O(|\bdry \Gamma|)$ bound on the rate at which it can be entangled with the rest of the system.

	\begin{lemma}
		\label{lem:depth-bound}
		Let $\rho_{RAX} \in \den{\hbt_{R} \otimes \hbt_{A} \otimes \hbt_{X}}$ be a state on $RAX$ of the form
		\[
		\rho_{RAX} = (\idchan_R \otimes \cL \circ (\cN_p^{\otimes m} \otimes\idchan_X)) (\sigma_{RAX} ) \ ,
		\]		
		for some arbitrary state $\sigma_{RAX}$ and $\cL \in \cptp{\hbt_{A} \otimes \hbt_{X}}$ be the quantum channel representing a quantum circuit of depth $\Delta$. Let $A' \subset A$ be an arbitrary subset of $A$ and
		assume that $\fid{\rho_{RA'}}{\xi_{RA'}} \geq 1-\delta$, where $\xi_{RA'} \in \den{\hbt_{R} \otimes \hbt_{A'}}$ is a pure state on $RA'$. Then for any $\Gamma \subset A$ qubits, we have
		\[
		3 \Delta |\bdry \Gamma| \geq  \ree{\Lambda}{\comp{\Lambda}}_{\xi} - \sqrt{{\delta}/{p^{|\Gamma|}}} |\Lambda| - g(\sqrt{{\delta}/{p^{|\Gamma|}}})
		\]
		where $\Lambda \equiv \Gamma \cap A'$, $\comp{\Lambda} \equiv A' \setminus \Lambda$.
	\end{lemma}
	\begin{proof}
		We write $\comp{\Gamma} = A \setminus \Gamma$, $\Lambda \equiv \Gamma \cap A'$ and $\comp{\Lambda} \equiv A' \setminus \Lambda$. 
		We can write $\rho_{RAX}$ as 
		\begin{align*}
			\rho_{RAX} = (\idchan_R \otimes \cL \circ (\cN_p^{\otimes m} \otimes\idchan_X)) (\sigma_{RAX} ) &=  p^{|\Gamma|}(\idchan_R \otimes \cL \circ (\cN_{\Gamma} \otimes\idchan_X)) (\sigma_{RAX}) \\ & + (1-p^{|\Gamma|})(\idchan_R \otimes \cL \circ (\cM_{\Gamma} \otimes\idchan_X)) (\sigma_{RAX})\
		\end{align*}
		where $\cN_{\Gamma} = (\idty/2^{|{\Gamma}|}\tr_{{\Gamma}})\otimes\cN_{p}^{\otimes |\comp{{\Gamma}}|}$, $\cM_{\Gamma} \in \cptp{\hbt_{A}}$ is some quantum channel.
		For the sake of readability, we write $\rho^{\cL \circ \cN_{\Gamma}} \equiv (\idchan_R \otimes \cL \circ \cN_{\Gamma}) (\sigma_{RAX})$. 
		We can now show, by using Lemma~\ref{lem:convex-is-close}, that the state that suffered complete erasure still has to end up close to the target state:
		\[
		\fid{\rho^{\cL \circ \cN_{\Gamma}}_{RA'}}{\xi_{RA'}} \leq 1 - {\delta}/{p^{|\Gamma|}} \ .
		\]
		From Properties 1 and 2 of Proposition~\ref{prop:ree-properties}, we are able to show that $\Gamma$ will be entangled with the rest, because $\Lambda$ is too, and this will lower bound the entanglement in $\rho^{\cL \circ \cN_{\Gamma}}$:
		\[
		\ree{\Gamma}{\comp{\Gamma}X}_{\rho^{\cL \circ \cN_{\Gamma}}} \geq \ree{\Lambda}{\comp{\Lambda}}_{\rho^{\cL \circ \cN_{\Gamma}}} \geq \ree{\Lambda}{\comp{\Lambda}}_{\xi} - \sqrt{{\delta}/{p^{|\Gamma|}}} |\Lambda| - g(\sqrt{{\delta}/{p^{|\Gamma|}}}) \ .
		\]
		On the other hand, from Lemma \ref{lem:small-incremental-entangling}, and the fact that $\cL$ has depth at most $\Delta$ we have $\ree{\Gamma}{\comp{\Gamma}X}_{\rho^{\cL \circ \cN_{\Gamma}}} \leq 3\Delta|\bdry \Gamma| + \ree{\Gamma}{\comp{\Gamma}X}_{\cN_{\Gamma}(\sigma)} = 3 \Delta|\bdry \Gamma|$.
		We therefore obtain
		\[
		3 \Delta |\bdry \Gamma| \geq  \ree{\Lambda}{\comp{\Lambda}}_{\xi} - \sqrt{{\delta}/{p^{|\Gamma|}}} |\Lambda| - g(\sqrt{{\delta}/{p^{|\Gamma|}}}) \ .
		\]
		
	\end{proof}
	
	\subsection{Overhead theorem}
	
	Here we prove our main result, which consists mainly in combining Lemma \ref{lem:structure-unitary}, and Lemma \ref{lem:depth-bound} harmoniously.
	
	\begin{theorem}
	\label{thm:main-overhead}
		Let $\cC$ be a quantum code encoding $k$ qubits into $n$ qubits. For an error correction module having a $D$-dimensional $O(1)$-local quantum circuit satisfying Definition \ref{def:ec-module}, with width $m$ and achieving a logical error rate $\delta$, we have
		\[
		\frac{m}{k} \in  \Omega\left( \min\left\{ \frac{1}{\Delta} \left(\frac{\log(1/\delta)}{\log(1/p)}\right)^{1/D}, \frac{1}{\delta^{1/8}} \right\} \right) \ .
		\]
	\end{theorem}
	We note that the $\Omega$ notation hides a constant that depends only on the dimension $D$. Observe that we are interested in the regime where the logical error rate $\delta$ goes to zero, $p$ is constant and $\Delta$ is constant. In this case, the bound becomes $\Omega\left(\log(1/\delta)^{1/D} \right)$. For an arbitrary connectivity graph $G$, we would partition $A$ into $\sim \log(1/\delta)$ sets of size $\sim \frac{m}{\log(1/\delta)}$ each having a boundary of size at most $|\bdry|$, and the bound would have the form $m/k \in \Omega(\log(1/\delta)/|\bdry|)$.
	
	\begin{proof}
		
		By the definition of the error correction module \ref{def:ec-module}, 
		we have 
		\begin{align}
		\label{eq:recov}
			F\left( \idchan_{R} \otimes \left( \tr_{\comp{A'}X} \circ [\fW]_{p} \circ (\cU \otimes \cP) \right)(\Phi_{RL}), 
			\idchan_{R} \otimes \cU(\Phi_{RL}) \right) \geq 1-\delta \ ,
		\end{align}
		
		
		Let $\{\Gamma_i\}_{i=1}^{\ell}$ be a partition of $A$, and $\{\Lambda_i\}_{i=1}^{\ell}$ where $\Lambda_i \equiv \Gamma_i \cap A'$ its induced partition on $A'$. We can apply Lemma \ref{lem:structure-unitary} by considering $\cU : \hbt_L \rightarrow \hbt_{A'}$ the encoding isometry of the code, and the existence of a recovery channel $\cD(\, \cdot \, )$ follows from~\eqref{eq:recov}. In fact $\cD$ is simply $\tr_{\comp{A'}X} \circ [\fW]_{p}$ without the first layer of noise. That gives
		\begin{equation}
			\label{eq:ree_lower_bound}
			\sum_i \ree{{\Lambda_i}}{\comp{\Lambda_i}}_{\idchan_R\otimes \cU (\Phi_{RL})} \geq   k - \sum_i 2\sqrt{\delta/{p^{|\Lambda_i|}}}|\AL| + g(\sqrt{\delta/{p^{|\Lambda_i|}}}) \ .
		\end{equation}
		
		Let $\rho_{RAX} = \idchan_{R} \otimes \left( [\fW]_{p} \circ (\cU \otimes \cP) \right)(\Phi_{RL})$ and $\xi_{RA'} = \idchan_R \otimes \cU(\Phi_{RL})$. The condition~\eqref{eq:recov} together with Lemma \ref{lem:depth-bound} implies that we have for all $i$
		\[
				3 \Delta |\bdry \Gamma_i| \geq  \ree{\Lambda_i}{\comp{\Lambda_i}}_{\xi} - \sqrt{{\delta}/{p^{|\Gamma_i|}}} |\Lambda_i| - g(\sqrt{{\delta}/{p^{|\Gamma_i|}}}) \ .
		\]
		This results in 
		\begin{align*}
			3\Delta \sum_i |\bdry \Gamma_i| &\geq   k - \sum_i 2\sqrt{\delta/{p^{|\Lambda_i|}}}|\AL| + g(\sqrt{\delta/{p^{|\Lambda_i|}}}) + \sqrt{\delta/{p^{|\Gamma_i|}}} |\Lambda_i| + g(\sqrt{\delta/{p^{|\Gamma_i|}}}) \ .
		\end{align*}	
		Since $|\Gamma_i| \geq |\Lambda_i|$, and using $f \equiv \log_p(\delta)$, the expression can be further simplified to
		\begin{align*}
			3\Delta \sum_i |\bdry \Gamma_i| &\geq k - \sum_i 3\sqrt{p^{f - |\Gamma_i|}} |\Gamma_i| + 2g(\sqrt{p^{f - |\Gamma_i|}}) \ .
		\end{align*}
		One can specialize the equation above to 
		\begin{align*}
		 3 \Delta \cdot \ell \cdot \max_i |\bdry \Gamma_i| 
		 &\geq k - \ell \cdot \max_i  3 p^{(f-|\Gamma_i|)/2} |\Gamma_i| +  4p^{(f - |\Gamma_i|)/4} \\ 
		 &\geq k - \ell \cdot \max_i  7 p^{(f-|\Gamma_i|)/4} |\Gamma_i|
		\end{align*}
		with $\ell$ the cardinality of the partition $\{\Gamma_i\}_{i=1}^{\ell}$ and using the fact that our choice of $\Gamma_i$ will be such that $f \geq |\Gamma_i|$.
		In $D$-dimensions, it is possible to find a partition $\{\Gamma_i\}_{i=1}^{\ell}$ such that $|\Gamma_i| \leq \lambda$, $|\bdry \Gamma_i| \leq c_1(D) \lambda^{(D-1)/D}$, and $\ell \leq c_2(D) m/\lambda$ for any $\lambda \geq 1$, where $c_1(D), c_2(D)$ are positive constants depending only on $D$, see Lemma \ref{lem:geometric-partitioning}. We take $\lambda = f/2$. With this choice, we have
		\begin{align*}
	7 \ell p^{(f-|\Gamma_i|)/4} |\Gamma_i| \leq 7 c_2(D) \frac{m}{f/2} \cdot p^{f/8} (f/2) = 7 c_2(D) m p^{f/8} \ . 
		\end{align*}
		On the other hand:
		\begin{align*}
		3 \Delta \cdot \ell \cdot \max_i |\bdry \Gamma_i| 
		&\leq 3 \Delta c_2(D) \frac{m}{f/2} c_1(D) (f/2)^{1 - 1/D} \\
		&= 3 c_1(D) c_2(D) \Delta m (f/2)^{-1/D} \ .
		\end{align*}
		Putting these together, we get 
		\begin{align*}
		m (3 c_1(D) c_2(D) \Delta (f/2)^{-1/D} + 7 c_2(D) p^{f/8}) \geq k  \ ,
		\end{align*}
		which leads to 
		\begin{align*}
		\frac{m}{k} \geq \frac{1}{2} \min \left(  \frac{f^{1/D}}{3 c_1(D)c_2(D)\Delta}, \frac{p^{f/8}}{7 c_2(D)} \right) \ .
				\end{align*}
%
%
	\end{proof}

	\section*{Acknowledgements}
	
	We would like to thank Cambyse Rouzé for discussions on quantum circuits in low dimension. We acknowledge funding from the European Research Council (ERC Grant AlgoQIP, Agreement No. 851716)
and from the Plan France 2030 through the project ANR-22-PETQ-0006.

	\bibliographystyle{alpha}
	\bibliography{references}

	\appendix

	\section{Additional lemmas}
	
	This appendix collects lemmas that are used in the proofs.
%
%
	
	\begin{lemma}
		\label{lem:convex-is-close}
		Let $\xi$ be a pure state, and $\rho = \lambda\rho_1 +(1-\lambda)\rho_2$ such that $F(\xi,\rho) \geq 1-\epsilon$, then $F(\xi, \rho_{1}) \geq 1-\epsilon/\lambda$.
	\end{lemma}
	\begin{proof}
	By writing $\xi = \proj{\xi}$, we have that
		\begin{align*}
			\fid{ \rho}{\xi} &\geq 1- \epsilon \\
			\bra{\xi} \rho \ket{\xi}&\geq 1- \epsilon \\	
			(1-\lambda)\bra{\xi} \rho_{2} \ket{\xi} + \lambda\bra{\xi} \rho_{1} \ket{\xi} &\geq 1- \epsilon \\
			(1-\lambda)+ \lambda\bra{\xi} \rho_{1} \ket{\xi} &\geq 1- \epsilon \\
			\bra{\xi} \rho_{1} \ket{\xi} &\geq 1- {\epsilon}/{\lambda} \ .
		\end{align*}
	\end{proof}

	\begin{lemma}[Idea taken from \cite{fawzi2015quantum}]
		\label{lem:approximate-markov}
		Let $\rho, \sigma \in \hbt_{ABC}$ such that 
		
		\[
		\sigma = \idchan_A \otimes \cR(\rho_{AB})
		\]
		
		for some $\cR \in \cptp{\hbt_{B}, \hbt_{BC}}$, with $\fid{\rho}{\sigma} \geq 1-\epsilon$
		\[
		\cond{A}{C}{B}_{\rho} \leq 2\sqrt{\epsilon} |A| + g(\sqrt{\epsilon}).
		\]
	\end{lemma}
	\begin{proof}
		We have
		
		\begin{align*}
			\cond{A}{C}{B}_{\rho} & = \coh{A}{BC}_{\rho} - \coh{A}{B}_{\rho_{AB}} \\ 
			&\leq \coh{A}{BC}_{\rho} - \coh{A}{BC}_{ \idchan_A \otimes \cR(\rho_{AB})} \\
			&= \coh{A}{BC}_{\rho} - \coh{A}{BC}_{\sigma} \\
			& \leq 2\sqrt{\epsilon} |A| + g(\sqrt{\epsilon}) \ .
		\end{align*}
		where the first inequality comes from the monotonicity of the coherent information, and the last comes from the continuity of the coherent information, see Proposition \ref{prop:coh-preperties}.
	\end{proof}
	
	\begin{lemma}
		\label{lem:coh-lower-bounds-ree}
		Let $\rho \in \den{\hbt_{A}\otimes \hbt_{B}}$, then 
		
		\[
		\coh{A}{B}_{\rho} \leq \ree{A}{B}_{\rho}
		\]
	\end{lemma}
	\begin{proof}
		Write $\sigma =\sum_i p_i \sigma_{A,i} \otimes \sigma_{B,i}, \sigma_{A,i} \in \hbt_{A}, \sigma_{B,i} \in \den{\hbt_{B}} $ the state for which 
		
		\[
		\ree{A}{B}_\rho = \relent{\rho}{\sigma} = \min_{\tau \in \sep{A}{B}} \relent{\rho}{\tau}
		\]
		
		Then we have $\relent{\rho}{\sigma} \geq \relent{\rho}{\idty \otimes (\sum_i p_i \sigma_{B,i})}$, as $\idty \otimes (\sum_i p_i \sigma_{B,i}) \geq \sigma$. This then gives 
		
		\[
		\ree{A}{B}_\rho  \geq \relent{\rho}{\sigma} \geq \relent{\rho}{\idty \otimes (\sum_i p_i \sigma_{B,i})} \geq \relent{\rho}{\idty \otimes \rho_B} = \coh{A}{B}_{\rho}
		\]
		
		where the first inequality is from the definition of the REE, the second is from Proposition 11.8.2 of \cite{wilde2017quantum}, and the last is from Equation 11.127 of \cite{wilde2017quantum}.
	\end{proof}

	\subsection{Geometric embeddings}

	\begin{definition}
		Let $G = (V,E)$ be a graph. Consider a partition $A \sqcup S \sqcup B$. Then $S$ is an $\alpha$-separator, $1/2 \leq \alpha < 1$, if 
		
		\begin{enumerate}
			\item $E$ contains no edges between $A$ and $B$
			\item $|A|,|B| \leq \alpha |G|$
		\end{enumerate}
		
	\end{definition}
	
	\begin{definition}
		A graph $G=(V,E)$ said to be $(f, \alpha)$-separable if every subgraph $G' \subset G$ has an $\alpha$-separator $S'$ with $|S'|\leq f(|G'|)$.
	\end{definition}
	
	\begin{lemma}
		
		\label{lem:geometric-partitioning}
		Let $G = (V,E)$, and $\eta: V \rightarrow \R^D$ such that 
		
		\begin{enumerate}
			\item $\forall u,v \in V,u \neq v , \| \eta(u) - \eta(v) \|_2 \geq 1$
			\item $\forall (u,v) \in E, \|\eta(u) - \eta(v)\|_2 \leq r$
		\end{enumerate}
		
		Then for any $\lambda \in \N$, there exists a partition $\{\Lambda_i\}_{i = 1}^{l}$ of $V$ such that $|\Lambda_i|\leq \lambda$, $|\bdry \Lambda_i| \in O(r \cdot \lambda^{(D-1)/D})$, $l \in O(|V|/\lambda) $ 
		
	\end{lemma}
	
	\begin{proof}
		
		From Corollary 3.8 \cite{miller1991unified}, $G$ is $(f,\alpha)$-separable, for  $f(n) \in O(r \cdot n^{(D-1)/D})$, and $\alpha = \frac{D+1}{D+2}$. Using Lemma 2 of \cite{henzinger1997faster}, this allows us to guarantee the existence of the partition $\{\Lambda_i\}_{i = 1}^{i=l}$ with the desired properties.
	\end{proof}

\end{document}